\documentclass[11pt]{article}
\usepackage{amsfonts}
\usepackage{adjustbox}
\usepackage{mathrsfs}
\usepackage[latin1]{inputenc}
\usepackage{float}
\usepackage{amsmath,amssymb}
\usepackage{graphicx}
\usepackage{caption}
\usepackage{subcaption}
\usepackage{amsfonts}
\usepackage{latexsym}
\usepackage{epstopdf}
\usepackage{geometry} 
\usepackage{bbm} 
\usepackage[active]{srcltx}
\usepackage{graphicx}
\usepackage{epstopdf}
\usepackage{booktabs}
\usepackage{multirow}
\usepackage{siunitx}
\usepackage{scrextend}
\usepackage{lipsum}
\usepackage{setspace}
\usepackage{amsmath, mathtools}
\usepackage[ruled,vlined]{algorithm2e}
\usepackage{enumerate}

\usepackage[
bookmarks=true,         
bookmarksnumbered=true, 
colorlinks=true, pdfstartview=FitV, linkcolor=blue, citecolor=blue,
urlcolor=blue]{hyperref}

\providecommand{\U}[1]{\protect\rule{.1in}{.1in}}
\newtheorem {theorem}{Theorem}[section]

\newtheorem {corollary}{Corollary}[section]

\newtheorem{example}{Example}[section]

\newtheorem{remark}{Remark}[section]

\newcommand{\E}{\mathbb{E}}
\newenvironment{proof}[1][Proof]{\textbf{#1.} }{\
\rule{0.5em}{0.5em}}

\newcommand{\bi}[1]{\mbox{\boldmath{$ #1 $}}}

\begin{document}

\title{Comparing Two Categorical Gini Correlations with Applications to Classification Problems}

\author{Sameera Hewage\textsuperscript{a}\thanks{CONTACT Sameera Hewage. Email: sameera.hewage@yahoo.com}\; and Yongli Sang\textsuperscript{b}}
\date{%
    \textsuperscript{a}Department of Physical Sciences \& Mathematics, West Liberty University, West Liberty, WV 26074, USA.\\%
    \textsuperscript{b}Department of Mathematics, University of Louisiana at Lafayette, Lafayette, LA 70504, USA.\\[2ex]%
    \today
}

\maketitle

\begin{abstract}

This article proposes an inferential framework for comparing predictor importance in classification problems with categorical response variables. The approach is based on the categorical Gini correlation (CGC) proposed by Dang \emph{et al}.~\cite{Dang2021}, a measure of dependence between numerical predictors and categorical outcomes. Predictor importance is evaluated by testing differences in CGCs across competing predictor groups.
The proposed methodology accommodates predictors of arbitrary and unequal dimensions and allows for dependence between predictor groups. Asymptotic normality of the test statistic is established under both the null and alternative hypotheses, and the resulting test is shown to be consistent. In addition to deriving the asymptotic distribution, a nonparametric bootstrap procedure is developed as an alternative approach to inference.
Simulation studies, along with applications to breast cancer and human activity recognition datasets, demonstrate the effectiveness of the proposed framework.

\end{abstract}
\noindent 


\section{Introduction}

Classification methods are widely used in applications where the response
variable is categorical, including disease diagnosis in medicine, species
identification in biology, risk assessment in finance, and customer
segmentation in marketing. In such settings, $n$ observations are assigned to
one of $K$ predefined classes based on numerical predictors. Beyond predictive
modeling, practitioners are often interested in assessing and comparing the
marginal strength of association between a categorical response and
predefined groups of predictors, such as biological pathways, sensor
modalities, or alternative feature representations. These questions commonly arise in applications where competing predictor sets are motivated by scientific knowledge, domain expertise, or experimental design, and often arise prior to model fitting or in situations where model assumptions are difficult to justify.

This article addresses these questions by developing statistical inference
procedures for comparing the categorical Gini correlations (CGCs) of different
predictor groups with respect to a common categorical response. 
The CGC ($\rho_g$) was proposed  in \cite{Dang2021}
to  measure  dependence between numerical and categorical variables. 
Let the categorical variable $Z$ take values
$L_1,\dots,L_K$ with $P(Z=L_k)=p_k>0$, and let the numerical random vector
$X\in\mathbb{R}^p$ follow distribution $F$. Suppose the conditional distribution
of $X$ given $Z=L_k$ is $F_k$ for $k=1,\dots,K$. Let $(X,X')$ and
$(X^{(k)},X^{(k)'})$ denote independent pairs drawn from $F$ and $F_k$,
respectively. The categorical Gini covariance and correlation between $X$ and
$Z$ are defined as
\begin{equation}\label{gcov}
\mathrm{gCov}(X,Z)=\Delta_F-\sum_{k=1}^K p_k \Delta_{k,F},
\end{equation}
and
\begin{equation}\label{mgc}
\rho_g(X,Z)=\frac{\Delta_F-\sum_{k=1}^K p_k \Delta_{k,F}}{\Delta_F},
\end{equation}
where $\Delta_F=\mathbb{E}\|X-X'\|$ and
$\Delta_{k,F}=\mathbb{E}\|X^{(k)}-X^{(k)'}\|$ denote the multivariate Gini mean
difference (GMD) under $F$ and $F_k$, respectively. The GMD was originally
introduced by Gini \cite{Gini1914} as an alternative measure of variability to
the standard deviation \cite{Yitzhaki2003, HewageGMD}. Consequently, the CGC admits a
natural interpretation as the proportion of total Gini variation explained by
between-group differences. It is known that $\rho_g(X,Z)=0$ if and only if $X$
and $Z$ are independent.
Owing to these properties, CGC provides a natural tool for assessing predictor
relevance in classification problems. Larger values of $\rho_g(X,Z)$ indicate
stronger marginal association between a predictor (or predictor group) and the
categorical response. This insight has motivated CGC-based feature screening
procedures for classification \cite{Sang2024}, where predictors are ranked
according to their sample CGCs. While such screening methods aim to identify
important variables from a large collection of candidates, they address a
different inferential objective from the present work. 

In this paper, we focus on formal statistical inference for comparing two population-level CGCs, $\rho_g(X,Z)$ and $\rho_g(Y,Z)$, where $X\in\mathbb{R}^p$ and $Y\in\mathbb{R}^q$ are two numerical predictor groups associated with a common categorical response $Z$. The dimensions $p$ and $q$ may be equal or different, and the predictor groups may be dependent.  We develop test statistics for comparing two CGCs and establish their asymptotic behavior under both the null and alternative hypotheses. We show that the proposed statistic admits a  normal limiting distribution. In addition, we propose a nonparametric bootstrap procedure that offers a flexible alternative for inference. Through simulation studies and real data analyses, we demonstrate that the proposed methods yield interpretable and robust inference across a broad range of classification settings.
When $p=q=1$, comparisons reduce to the setting of individual predictors, and
related screening approaches have been studied in
\cite{Cui15, Ni16, Cheng17, Lai17, He19}. However, in many modern applications,
predictors naturally appear in groups. For example, in omics studies
\cite{Buch2021}, researchers often investigate whether gene pathways or
functionally related gene sets are associated with disease outcomes. Grouped
predictor comparisons are also relevant in areas such as image recognition,
sensor-based activity classification, disease prediction, and credit risk
analysis. A range of grouped feature selection methods has been proposed in the
literature \cite{Yuan2006, Meier2008, Belhechmi2020, Niu2020, Buch2021,
Wang23}, including CGC-based approaches.

From a methodological perspective, comparing correlations of different
predictors with a common response has a long history in the statistics
literature. Classical procedures such as Hotelling's $t$ \cite{Hotelling1940},
Williams' $t$ \cite{Williams1959a}, Olkin's $z$ \cite{Olkin1967}, and MMR's $z$
\cite{Meng1992} focus on comparing dependent correlation coefficients when both
predictors and responses are numerical. In contrast, the present work addresses
the setting where the response variable is categorical, and the predictors may
be multivariate and dependent. 

Moreover, the
proposed framework naturally extends to assessing whether augmenting a
baseline predictor group with additional variables provides incremental
marginal information for a classification task, offering a complementary
perspective to likelihood ratio tests in multinomial regression. Unlike
regression-based approaches, which assess conditional predictive contributions
within a specified parametric model, CGC-based comparison focuses on
quantifying and contrasting intrinsic marginal dependence between predictor
groups and the response in a model-free manner. This perspective is
particularly useful when predictors are strongly dependent, when
regression-based inference may be unstable or difficult to interpret, or when
predictors are defined a priori by scientific or experimental
considerations. As such, the contributions of this paper are:
\begin{itemize}
    \item Developing a statistical framework to compare the marginal association of two predictor groups with a categorical outcome.
\item Extending the framework to assess whether adding predictors to a baseline group provides additional information for classification.
\end{itemize}

We use the following notations throughout this paper.  $\|\cdot\|$ denotes the Euclidean norm, defined as $\| \bi x\|=\sqrt{x^2_1+x^2_2+\cdots+x^2_p}$ for a $p$-vector $\bi x=(x_1, x_2, \cdots, x_p)^T \in \mathbb{R}^p$. For two  sequences of real numbers, $a_n, b_n$, the notation $a_n=o(b_n)$ means $\lim_{n \to \infty}{a_n}/{b_n}=0$, while $a_n=O(b_n)$ means that $L \leq {a_n}/{b_n} \leq U$ for some finite constants $L$ and $U$. For sequences of random variables, the notations $o_P(n)$ and $O_P(n)$ similarly describe the relationships that hold in probability. 

The remainder of the paper is organized as follows. Section~\ref{sec:cgc}
reviews sample CGCs and presents the proposed comparison procedures. Section
\ref{sec:simulationstudy} reports simulation studies under independent and
dependent predictor scenarios. Section~\ref{sec:realdata} presents real data
applications. Concluding remarks are given in Section~\ref{sec:conclusion}, and
all technical proofs are provided in the Appendix.

  \section{Comparing two categorical Gini correlations}\label{sec:cgc}
 Let $W=(X, Y)^T$ be a non-degenerate random vector from $H \in \mathbb{R}^{p+q}$ with marginal distributions   $F \in \mathbb{R}^p$ and $G \in\mathbb{R}^q$, respectively. 
The purpose of this section is to  compare the importance of $X$ and $Y$ to $Z$ by conducting  statistical tests. 
 
 We assume that the conditional distribution of $Y$ given $Z=L_k$ is $G_k$. Analogous to (\ref{gcov}) and (\ref{mgc}), 
 the categorical Gini covariance and correlation between $Y$ and $Z$ can be defined  by
\begin{equation*}\label{gcov2}
\mbox{gCov}(Y,Z) = \Delta_G-\sum_{k=1}^Kp_k\Delta_{k, G},
\end{equation*}
and 
\begin{equation} \label{mgc2}
\rho_g(Y, Z) = \frac{ \Delta_G-\sum_{k=1}^Kp_k\Delta_{k,G}}{\Delta_G},
\end{equation}
where
$\Delta_G=\E\|Y-Y'\|,$ $ \Delta_{k, G}=\E \|Y^{(k)}-Y^{(k)'}\|$ with $(Y, Y')$ and $(Y^{(k)}, Y^{(k)'})$ be independent pair variables from $G$ and $G_k$, respectively.

The hypothesis test of interest is 
\begin{align}\label{test:H0}
 \mathcal{H}_0: \rho_g(X, Z)=\rho_g(Y, Z) \ \ \text{vs.}  \ \  \mathcal{H}_1: \rho_g(X, Z) > \rho_g(Y, Z).
 \end{align}
For the simplicity of notation, we let $\rho_1$ and $\rho_2$ represent for $\rho_g(X, Z)$ and $\rho_g(Y, Z)$, respectively.

Suppose a sample ${\cal W} = \{(X_1, Y_1, Z_1), (X_2, Y_2, Z_2), \ldots, (X_n, Y_n, Z_n)\}$ is drawn from the joint distribution of $(X, Y)$ and $Z$. We can write ${\cal W} = {\cal W}_1 \cup {\cal W}_2 \cup \ldots \cup {\cal W}_K$, where ${\cal W}_k = \left\{ \big(X^{(k)}_{1}, Y^{(k)}_{1}\big), \big(X^{(k)}_{2}, Y^{(k)}_{2}\big), \ldots, \big(X^{(k)}_{n}, Y^{(k)}_{n}\big) \right\}$ are the samples with $Z_i=L_k$ and $n_k$ are the numbers of sample points in the $k^{th}$ class.  
Then the Gini correlations can be respectively estimated by (\cite{Sang2023},  \cite{Hewage}, \cite{HewageCGC})
\begin{align}\label{sgcor}
&\hat{\rho}_1=\dfrac{\textrm{gCov}_n(X, Z)}{\widehat{\Delta_F}}, \  \hat{\rho}_2=\dfrac{\textrm{gCov}_n(Y, Z)}{\widehat{\Delta_G}},
\end{align}
where 
$\hat{p}_k=\dfrac{n_k}{n}$,
\begin{align}\label{sgcov}
&\textrm{gCov}_n(X, Z)= {n \choose 2}^{-1}\sum_{1 \leq i < j \leq n}\|X_i-X_j\|-\sum_{k=1}^K \hat{p}_k {n_k \choose 2}^{-1}\sum_{1 \leq i < j \leq n_k}\|X^{(k)}_i-X^{(k)}_j\|, \nonumber\\
&\widehat{\Delta}_F= {n \choose 2}^{-1}\sum_{1 \leq i < j \leq n}\|X_i-X_j\|, \nonumber \\ 
&\textrm{gCov}_n(Y, Z)= {n \choose 2}^{-1}\sum_{1 \leq i < j \leq n}\| Y_i-Y_j\|-\sum_{k=1}^K \hat{p}_k {n_k \choose 2}^{-1}\sum_{1 \leq i < j \leq n_k}\|Y^{(k)}_i- Y^{(k)}_j\|,\nonumber\\
& \widehat{\Delta}_G={n \choose 2}^{-1}\sum_{1 \leq i < j \leq n}\|Y_i-Y_j\|.
\end{align}
Then the test statistic for (\ref{test:H0}) is 
\begin{align*}
D_{n}=\hat{\rho}_1-\hat{\rho_2}.
\end{align*}

We study  the limiting distributions of the proposed test statistic, $D_{n}$,  under both the null and alternative hypotheses in Section \ref{sec:limit}.  
 In addition, we develop a nonparametric bootstrap procedure in Section \ref{sec:boot} as an alternative approach for inference.

\subsection{Asymptotic Normality of the CGC Difference} \label{sec:limit}
We first state the conditions that we impose to derive the limiting distributions of $D_n$ to facilitate the tests:
\begin{description} \label{Condition}
\item \textbf{C}1. $\E \| X\|^2 < \infty$, $\E \| Y\|^2 < \infty$;
\item \textbf{C}2. $\dfrac{n_k}{n} \to p_k>0,$  $p_1+p_2+...+p_K=1$.
\item \textbf{C}3. $\max\{\rho_1, \rho_2\} >0$.
\item \textbf{C}4. $P(Y = aX + b) < 1$ for all real constants $a$ and $b$.
\end{description}

Condition \textbf{C}2 indicates that none of the classes from the $K$ groups can dominate the others. 
We assume at least one of the two Gini correlations is non-zero in condition \textbf{C}3. If both correlations are zero, then neither predictor is significant to $Z$, and there is no need to compare the significance of the covariates to the categorical response. In order to check  condition \textbf{C}3, before conducting a comparison, one can test if each correlation is zero using the permutation method \cite{Dang2021}. Condition \textbf{C}4 rules out the degenerate case in which $X$ and $Y$ are perfectly linearly dependent, that is, $Y = aX + b$ almost surely. This corresponds to perfect collinearity and a singular covariance structure, where the two predictors contain no distinct information. Such exact linear dependence is uncommon in practice due to noise and natural variability, making this a mild and reasonable assumption.

\begin{theorem}\label{wilkrho2}
Under conditions \textbf{C}1- \textbf{C}4, and $\mathcal{H}_0$,  as $n \to \infty$, we have 
\begin{align}\label{asymp_ind}
 \dfrac{ \hat{\rho}_1 - \hat{\rho}_2}{\sigma_0}\stackrel{d}{\rightarrow} \mathcal{N}(0,1),\end{align}
where $\sigma^2_0=\textrm{var}(\hat{\rho}_1-\hat{\rho}_2)$.
\end{theorem}

The proof of Theorem~\ref{wilkrho2} is provided in the Appendix for interested readers. In order to apply the asymptotic normality in (\ref{asymp_ind}) to make inference, we need to find  a consistent estimator for $\sigma^2_0$. 
Suggested by the authors of \cite{Dang2021}, we will estimate $\sigma^2_0$ by the jackknife method \cite{Shao1996}. 
Let $\hat{\delta}^{(-i)}_{n-1}$  be the estimator of $\rho_1-\rho_2$, based on the sample with the $i^\text{th}$ observation deleted. Then the jackknife estimator of variance, $\sigma^2_0$, is given by
\begin{align*}
\widehat{M}=\dfrac{n-1}{n}\sum_{i=1}^n \Bigl(\hat{\delta}^{(-i)}_{n-1} - \bar{\hat{\delta}}_{(\cdot)} \Bigr)^2,
\end{align*}
where $\bar{\hat{\delta}}_{(\cdot)} =\dfrac{1}{n} \sum_{i=1}^n {\hat{\delta}^{(-i)}_{n-1}}$. 

 Then ${\cal H}_0$ will be rejected if $\hat{\rho}_1-\hat{\rho}_2>Z_{\alpha}\sqrt{\widehat{M}}$ at level $\alpha$, where $Z_{\alpha}$ is the $(1-\alpha)100\%$ quantile of the standard normal variable.

The power function of the test can be computed by
\begin{align*}
P_n(\alpha) = P_{\mathcal{H}_1} \left(\hat{\rho}_1-\hat{\rho}_2  > Z_{\alpha}\sqrt{\widehat{M}}\right).
\end{align*}


%

\begin{theorem}\label{wilkrho3}
Under conditions \textbf{C}1- \textbf{C}3, and alternative hypothesis $\mathcal{H}_1$,   we have 
$$ \dfrac{ (\hat{\rho}_1 - \hat{\rho}_2) - (\rho_1 - \rho_2)}{\sigma_1}\stackrel{d}{\rightarrow} \mathcal{N}(0,1),$$
where $\sigma^2_1$ is the asymptotic variance defined by (\ref{sig:h1}) in the proof.
\end{theorem}

The consistency of the test is established as a consequence of Theorem \ref{wilkrho3}. 
\begin{corollary}\label{power2}
For any alternative $\mathcal{H}_1$, under conditions \textbf{C}1- \textbf{C}3, as $n \to \infty$,
$$P_{\mathcal{H}_1}\left(\hat{\rho}_1-\hat{\rho}_2   > Z_{\alpha}  \sqrt{\hat{M}} \right) \to 1.$$
\end{corollary}

\begin{remark}\label{remark:1}
From Theorem \ref{wilkrho2}, the proposed test statistic admits a null asymptotic normal distribution provided that the two predictors are not perfectly linearly correlated. In particular, this condition encompasses the case where $X$ and $Y$ are independent.
\end{remark}

\begin{remark}\label{remark:2}
In classification tasks, one typically predicts \(Z\) using both \(X\) and \(Y\). A more relevant question is whether the variables in \(Y\) provide additional useful information for classification beyond \(X\). This can be formally assessed using the proposed approach as follows:
\begin{enumerate}
    \item Define $\rho_1 = \rho_g(W, Z)$ and $\rho_2 = \rho_g(X, Z)$, where $W = \begin{pmatrix} X \\ Y \end{pmatrix} \in \mathbb{R}^{p+q}$.
    \item Apply Theorem~\ref{wilkrho2} to test the null hypothesis $\mathcal{H}_0: \rho_1 = \rho_2$ against the one-sided alternative $\mathcal{H}_a: \rho_1 > \rho_2$.
\item A small $p$-value provides evidence that including $Y$ offers additional useful information for the classification task.
\end{enumerate}
\end{remark}

%

\subsection{Nonparametric Bootstrap Procedure}\label{sec:boot}
In this section, we propose a nonparametric bootstrap method to compare the importance 
of covariates with respect to the same categorical response, without requiring any 
assumptions about dependence or independence between the covariates.

\begin{enumerate}[Step 1:]
\item For the $n$ i.i.d.\ observations
$
{\cal W}=\left\{\left(\begin{pmatrix}X_1\\Y_1\end{pmatrix},Z_1\right),
\left(\begin{pmatrix}X_2\\Y_2\end{pmatrix},Z_2\right),
\ldots,
\left(\begin{pmatrix}X_n\\Y_n\end{pmatrix},Z_n\right)\right\},
$
calculate the observed difference: $d_0=\hat{\rho}_1-\hat{\rho}_2$.

\item For $b = 1, \ldots, B$, generate bootstrap samples by resampling with replacement 
\emph{within each class separately}: for each $k = 1, \ldots, K$, draw a bootstrap 
sample $\mathcal{W}_k^{(b)}$ of size $n_k$ with replacement from the original 
class-specific data $\mathcal{W}_k$. For every $b$, calculate $\hat{\rho}_{1,b}$, 
$\hat{\rho}_{2,b}$, and the difference $u_b = \hat{\rho}_{1,b} - \hat{\rho}_{2,b}$, 
$b = 1, \ldots, B$.

\item Calculate the test statistic $d_b = u_b - \bar{u}_B$, $b = 1, \ldots, B$, 
where $\bar{u}_B = \dfrac{\sum_{b=1}^B u_b}{B}$.

\item Estimate the p-value as $\text{p-val} = \dfrac{1}{B} \sum_{b=1}^B I(d_b \geq 
d_{0})$ where $I(\cdot)$ is the indicator function.
\end{enumerate}

\begin{remark}\label{remark:3}
It should be noted that since resampling is performed within each class separately, 
every bootstrap replication is guaranteed to contain exactly $n_k \geq 1$ observations 
from class $L_k$, eliminating the possibility of empty categories.
\end{remark}

 \section{Simulation Studies}\label{sec:simulationstudy}
 
In this section, we present extensive simulation studies to assess the empirical performance of our proposed inference methods.  These simulations are designed to evaluate both Type I error control and statistical power under a range of data-generating scenarios. To the best of our knowledge, no existing methods are available for comparing predictor significance via correlation in classification. Therefore, we focus exclusively on evaluating the proposed methods.

In Examples \ref{exam:independent} and \ref{exam:dependent}, we consider settings in which the numerical predictors are either univariate or multivariate, while the categorical response variable has three levels ($K=3$). 
We investigate three different configuration of class sizes: balanced \( \bi n =(n_1, n_2, n_3)= (40, 40, 40) \), slightly unbalanced  \( \bi n = (50, 40, 30) \), and heavily unbalanced \( \bi n = (72, 36, 12) \). 
For each test, we conduct $3,000$ simulations and report the Type I error and power at a significance level of $\alpha=0.05$.
In Example \ref{exam:independent}, the predictors are independent while Example \ref{exam:dependent} addresses the cases where the predictors are dependent.

In Example \ref{exam:blrm}, we extend our evaluation to a binary logistic regression model. In part (a),  we apply our proposed tests to assess the significance of predictors with respect to a binary outcome. And we also assess whether augmenting a
baseline predictor group with additional variables provides incremental
marginal information for a classification task in part (b). 
$p$-values are computed based on $3,000$ simulation replications for each scenario. 

Across all examples,  we evaluate the performance of 
 \begin{labeling}[~-]{Second}
\item[\textbf{asN}] method based on the limiting distribution given in (\ref{wilkrho3});
\item[\textbf{Bootstrap}] bootstrap method in section (\ref{sec:boot}) with $B=1000$.
\end{labeling}
These evaluations allow us to identify the strengths and limitations of each method under varying distributional assumptions, sample sizes, and dependency structures.

\begin{example}\label{exam:independent}(Independent predictors)
Define 
\begin{align*}
& \bi \Sigma_{p \times p} =( \Sigma_{ij} )\in \mathbb{R}^{p \times p}, \quad \bi \Sigma_{q \times q} = (\Sigma_{ij}) \in \mathbb{R}^{q \times q} \ \text{with} \ \ \Sigma_{ij} = 0.7^{|i-j|};\\
&\bi \beta_p=(\beta,..., \beta)^T_p \ \ \text{where} \ \  \beta \in [0, 1] \ \text{represents the difference in distribution.}
\end{align*}
We generate samples of $X \in \mathbb{R}^p$ and $Y \in \mathbb{R}^q$ independently from mixture of normal and exponential distributions as below:
\begin{enumerate}[{\bf  (a).}]
 \item 
 \begin{align*}
& X_1 \sim N(\bi 2_p,\bi \Sigma_{p \times p}), X_2 \sim N\big(\bi 3_p, \bi \Sigma_{p \times p}\big)+\bi \beta_p, X_3  \sim N\big(\bi 4_p, \bi \Sigma_{p \times p}\big)+2\bi \beta_p;\\
 &Y_1 \sim N(\bi 0_q,\bi \Sigma_{q \times q}), Y_2 \sim N\big(\bi 1_q, \bi \Sigma_{q \times q}\big), Y_3  \sim N\big(\bi 2_q, \bi \Sigma_{q \times q}\big),
 \end{align*}
 with $(p,q)=(1,1), (5,5), (5, 10)$.
\item  
\begin{align*}
&X_1 \sim  \bi \Sigma^{1/2}_{p \times p}\exp(\bi 1)+\bi \beta_p, X_2 \sim  \bi \Sigma^{1/2}_{p \times p} \exp(\bi 2)+\bi \beta_p, X_3  \sim  \bi \Sigma^{1/2}_{p \times p}\exp(\bi 4),\\
&Y_1 \sim  \bi \Sigma^{1/2}_{p \times p}  \exp(\bi 1), Y_2 \sim  \bi \Sigma^{1/2}_{p \times p} \exp(\bi 2), Y_3  \sim  \bi \Sigma^{1/2}_{p \times p} \exp(\bi 4),
\end{align*}
with $(p,q)=(1,1), (5,5), (10, 10)$.
\end{enumerate}

Here \( \beta \in [0, 1] \) where with step size 0.2 
represents the differences in means of  the 3 populations of $X$. 
When $\beta=0$, the distributions of $X$ and $Y$ differ only in location. Both share the same relationship with the categorical variable.  Consequently, $\rho_1=\rho_2$. 
As $\beta$ increases, the conditional distributions $G_k, k=1,2,3$, of  $Y$ given $Z=L_k$ remain unchanged, while the differences among the conditional distributions of $X$ increase. This implies that the association between $X$ and $Z$ becomes stronger than that between $Y$ and $Z$.  
\end{example}


Type I error and power for each scenario are computed and reported in Tables \ref{tab:k3norm} and \ref{tab:k3exp} for Case a and Case b, respectively. 
We also compute the average value of $\hat{\rho}_1-\hat{\rho}_2$ in parentheses. The column  $\beta=0.0$ corresponds to the Type I error, while the other columns represent the statistical power. 

As expected,  the average value of $\hat{\rho}_1-\hat{\rho}_2$ increases as $\beta$ increases. Both methods perform similarly, effectively controlling the Type I error and demonstrating satisfactory power across all cases. The power of both methods increases as the difference in Gini correlations grows. 
Both methods are robust in unbalanced scenarios where sample sizes differ significantly (e.g. \bi n=(72, 36, 12)). They exhibit great power when $(p,q)=(5,5)$ compared to $(p,q)=(1,1)$.
In the case of exponential mixtures, both methods continue to control the Type I error effectively and maintain satisfactory power.  Additionally, they prove to be more effective in higher-dimensional settings.

\begin{table}[] 
\centering
\renewcommand{\arraystretch}{1.2}
\begin{tabular}{l  cc|c c c c c c} \hline\hline
 $(p, q)$ & $\bi n$ & method & $\beta=0.0$ &$\beta=0.2$ &$\beta=0.4$ &$\beta=0.6$ &$\beta=0.8$ &$\beta=1.0$ \\ \hline
 \multirow{9}*{ (1, 1)}&& &(.0001)&(.0617)&(.1197)&(.1710)&(.2172)&(.2563)\\
 & \multirow{2}*{(40,40,40)}& asN &.0543 &.2670&.6430&.8993&.9860&.9983\\
&&Bootstrap& .0603&.2853&.6627&.9080&.9880&.9980\\\cline{2-9}
&& &(-.0012)&(.0595)&(.1195)&(.1679)&(.2131)&(.2528)\\
& \multirow{2}*{(50,40,30)}& asN &.0473 &.2603&.6443&.8923&.9813&.9983\\
&&Bootstrap& .0523&.2740&.6580&.9033&.9840&.9987\\\cline{2-9}
&& &(.0006)&(.0507)&(.1007)&(.1463)&(.1884)&(.2312)\\
& \multirow{2}*{(72,36,12)}& asN &.0537 &.2167&.5283&.8033&.9447&.9910\\
&&Bootstrap& .0600&.2440&.5607&.8243&.9537&.9923\\\hline

 \multirow{9}*{ (5, 5)}&& &(-.0000)&(.0568)&(.1114)&(.1602)&(.2050)&(.2436)\\
 & \multirow{2}*{(40,40,40)}& asN &.0530 &.3970&.8687&.9910&.9997&1.000\\
&&Bootstrap& .0577&.4187&.8777&.9943&.9997&1.000\\\cline{2-9}
&& &(.0006)&(.0554)&(.1091)&(.1577)&(.2015)&(.2393)\\
& \multirow{2}*{(50,40,30)}& asN &.0530 &.3850&.8477&.9900&.9997&1.000\\
&&Bootstrap& .0583&.4000&.8583&.9920&.9997&1.000\\\cline{2-9}
&& &(.0003)&(.0472)&(.0914)&(.1340)&(.1751)&(.2135)\\
& \multirow{2}*{(72,36,12)}& asN &.0500 &.3210&.7347&.9490&.9967&1.000\\
&&Bootstrap& .0590&.3517&.7643&.9577&.9973&1.000\\\hline

 \multirow{9}*{ (5, 10)}&& &(.0014)&(.0601)&(.1117)&(.1618)&(.2061)&(.2459)\\
 & \multirow{2}*{(40,40,40)}& asN &.0543 &.4903&.9180&.9977&1.000&1.000\\
&&Bootstrap& .0623&.5140&.9237&.9980&1.000&1.000\\\cline{2-9}
&& &(.0003)&(.0576)&(.1101)&(.1594)&(.2023)&(.2424)\\
& \multirow{2}*{(50,40,30)}& asN &.0530 &.4647&.9100&.9973&1.000&1.000\\
&&Bootstrap& .0580&.4887&.9177&.9977&1.000&1.000\\\cline{2-9}
&& &(.0020)&(.0464)&(.0916)&(.1336)&(.1754)&(.2129)\\
& \multirow{2}*{(72,36,12)}& asN &.0543 &.3657&.8160&.9797&.9993&1.000\\
&&Bootstrap& .0627&.3937&.8383&.9830&.9993&1.000\\\hline

  \hline \hline 
\end{tabular}
\caption{Size and Power of tests for Example \ref{exam:independent} (a). }
\label{tab:k3norm}
\end{table}

\begin{table}[] 
\centering
\renewcommand{\arraystretch}{1.2}
\begin{tabular}{l  cc|c c c c c c} \hline\hline
 $(p, q)$ & $\bi n$ & method & $\beta=0.0$ &$\beta=0.2$ &$\beta=0.4$ &$\beta=0.6$ &$\beta=0.8$ &$\beta=1.0$ \\ \hline
 \multirow{9}*{ (1, 1)}&& &(.0006)&(.0712)&(.1469)&(.2146)&(.2753)&(.3292)\\
 & \multirow{2}*{(40,40,40)}& asN &.0473 &.3827&.8533&.9840&.9993&1.000\\
&&Bootstrap& .0473&.4043&.8697&.9887&.9993&1.000\\\cline{2-9}
&& &(-.0014)&(.0551)&(.1209)&(.1784)&(.2322)&(.2789)\\
& \multirow{2}*{(50,40,30)}& asN &.0413 &.3077&.7910&.9653&.9977&1.000\\
&&Bootstrap& .0480&.3383&.8213&.9753&.9997&1.000\\\cline{2-9}
&& &(-.0005)&(.0240)&(.0542)&(.0836)&(.1143)&(.1401)\\
& \multirow{2}*{(72,36,12)}& asN &.0420 &.1540&.3913&.6293&.8447&.9367\\
&&Bootstrap& .0520&.1883&.4727&.7460&.9207&.9747\\\hline

 \multirow{9}*{ (5, 5)}&& &(-.0001)&(.0505)&(.0988)&(.1436)&(.1845)&(.2201)\\
 & \multirow{2}*{(40,40,40)}& asN &.0500 &.3550&.8120&.9863&1.000&1.000\\
&&Bootstrap& .0550&.3757&.8320&.9897&1.000&1.000\\\cline{2-9}
&& &(-.0001)&(.0417)&(.0840)&(.1205)&(.1569)&(.1881)\\
& \multirow{2}*{(50,40,30)}& asN &.0527 &.2907&.7267&.9550&.9980&1.000\\
&&Bootstrap& .0583&.3120&.7563&.9680&.9983&1.000\\\cline{2-9}
&& &(-.0000)&(.0203)&(.0398)&(.0604)&(.0808)&(.0982)\\
& \multirow{2}*{(72,36,12)}& asN &.0500 &.1493&.3103&.5390&.7643&.8746\\
&&Bootstrap& .0560&.1697&.3560&.6210&.8440&.9380\\\hline

 \multirow{9}*{ (10, 10)}&& &(-.0000)&(.0451)&(.0846)&(.1231)&(.1574)&(.1888)\\
 & \multirow{2}*{(40,40,40)}& asN &.0507 &.4270&.8883&.9963&1.000&1.000\\
&&Bootstrap& .0553&.4530&.9007&.9983&1.000&1.000\\\cline{2-9}
&& &(.0003)&(.0357)&(.0715)&(.1043)&(.1342)&(.1623)\\
& \multirow{2}*{(50,40,30)}& asN &.0480 &.3400&.8053&.9810&.9997&1.000\\
&&Bootstrap& .0520&.3587&.8283&.9850&1.000&1.000\\\cline{2-9}
&& &(-.0002)&(.0182)&(.0350)&(.0532)&(.0712)&(.0863)\\
& \multirow{2}*{(72,36,12)}& asN &.0470 &.1563&.3487&.5957&.8170&.9197\\
&&Bootstrap& .0553&.1770&.3960&.6730&.8853&.9683\\\hline

  \hline \hline 
\end{tabular}
\caption{Size and Power of tests for Example \ref{exam:independent}(b). }
\label{tab:k3exp}
\end{table}


\begin{example}\label{exam:dependent} (Dependent predictors)
\begin{enumerate}[{\bf  (a).}]
\item  Let 
\begin{align*}
&\boldsymbol{\mu}_1 = (\mathbf{2}_{p}, \mathbf{0}_{q}), \  
\boldsymbol{\mu}_2 = (\mathbf{3}_{p}+\bi \beta_p, \mathbf{1}_{q}), \ 
\boldsymbol{\mu}_3 = (\mathbf{4}_{p}+2\bi \beta_p, \mathbf{2}_{q});\\
& \boldsymbol{\Sigma} = (\Sigma_{ij}) \in \mathbb{R}^{(p+q) \times (p+q)}, \text{ where } \Sigma_{ij} = 0.7^{|i-j|}.
\end{align*}
where $\beta \in [0,1]$ with step size 0.2 denotes  for the difference between means.
Generate samples of 
\[
\begin{pmatrix} X \\ Y \end{pmatrix}^{(1)} \sim \mathcal{N}_{p+q}(\boldsymbol{\mu}_1, \boldsymbol{\Sigma}), \quad
\begin{pmatrix} X \\ Y \end{pmatrix}^{(2)} \sim \mathcal{N}_{p+q}(\boldsymbol{\mu}_2, \boldsymbol{\Sigma}), \quad
\begin{pmatrix} X \\ Y \end{pmatrix}^{(3)} \sim \mathcal{N}_{p+q}(\boldsymbol{\mu}_3, \boldsymbol{\Sigma}).
\]
By the structure of the covariance matrix, $\Sigma$, the two predictors, $X$ and $Y$,  are dependent but not linearly. Therefore, Theorem \ref{wilkrho2} applies. When $\beta=0$, the distributions of $X$ and $Y$ differ only in location. Both share the same relationship with the categorical variable.  Consequently, $\rho_1=\rho_2$. 
As $\beta$ increases, the conditional distributions $G_k, k=1,2,3$, of  $Y$ given $Z=L_k$ remain unchanged, while the differences among the conditional distributions of $X$ increases. This implies that the association between $X$ and $Z$ becomes stronger  than that between $Y$ and $Z$.  We compute Type I error and power under various settings, and report the results in Table \ref{tab:indKnorm}. 
\item  Let \(\mathbf{\bi \omega}_1 = ( \omega_{11},  \omega_{12})^T\), where   $\omega_{11},  \omega_{12}$ are i.i.d. from \(\text{Exp}(1)\);
 
  \(\mathbf{\bi \omega}_2 = ( \omega_{21},  \omega_{22})^T\), where   $\omega_{21},  \omega_{22}$ are  i.i.d. from \(\text{Exp}(2)\);
  
   \(\mathbf{\bi \omega}_3 = ( \omega_{31},  \omega_{32})^T\), where $\omega_{31},  \omega_{32}$
   are i.i.d. from \(\text{Exp}(4)\).
   
Generate samples as follows:
\[
\begin{pmatrix} X \\ Y \end{pmatrix}^{(1)} \sim \Sigma^{1/2} \bi{ \omega}_1+( \beta,  0), \quad
\begin{pmatrix} X \\ Y \end{pmatrix}^{(2)} \sim \Sigma^{1/2} \bi{ \omega}_2+(\beta,  0), \quad
\begin{pmatrix} X \\ Y \end{pmatrix}^{(3)} \sim \Sigma^{1/2} \bi{ \omega}_3.
\]
\end{enumerate}
\end{example}
The empirical values are summarized in Figure \ref{fig:Rexp}.

\begin{table}[] 
\centering
\renewcommand{\arraystretch}{1.2}
\begin{tabular}{l  cc|c c c c c c} \hline\hline
 $(p, q)$ & $\bi n$ & method & $\beta=0.0$ &$\beta=0.2$ &$\beta=0.4$ &$\beta=0.6$ &$\beta=0.8$ &$\beta=1.0$ \\ \hline 
 \multirow{9}*{ (1, 1)}& &&(-.0004)&(.0620)&(.1190)&(.1714)&(.2176)&(.2571)\\
 & \multirow{2}*{(40,40,40)}& asN &.0517 &.5103&.9430&.9973&1.000&1.000\\
&&Bootstrap& .0553&.5223&.9477&.9973&1.000&1.000\\\cline{2-9}
&& &(.0009)&(.0611)&(.1171)&(.1683)&(.2149)&(.2555)\\
&\multirow{2}*{(50,40,30)}& asN &.0533 &.5107&.9447&.9987&1.000&1.000\\
&&Bootstrap& .0570&.5230&.9477&.9987&1.000&1.000\\\cline{2-9}
&& &(.0001)&(.0504)&(.1002)&(.1455)&(.1886)&(.2300)\\
& \multirow{2}*{(72,36,12)}& asN &.0463&.3907&.8773&.9907&.9990&1.000\\
&&Bootstrap& .0533&.4233&.8923&.9930&.9990&1.000\\\hline

\multirow{9}*{ (5, 5)}&& &(.00004)&(.0570)&(.1107)&(.1599)&(.2032)&(.2434)\\
 & \multirow{2}*{(40,40,40)}& asN &.0510 &.5050&.9430&.9997&1.000&1.000\\
&&Bootstrap& .0570&.5217&.9503&.9997&1.000&1.000\\\cline{2-9}
&& &(-.0011)&(0.0545)&(0.1076)&(0.1565)&(0.2008)&(0.2397)\\
& \multirow{2}*{(50,40,30)}& asN &.0447&.4780 &.9320 &.9990 &1.000&1.000\\
&&Bootstrap& .0480 &.4997&.9393 &.9990 &1.000&1.000\\\cline{2-9}
&& &(-.0001)&(.0449)&(.0917)&(.1339)&(.1740)&(.2119)\\
& \multirow{2}*{(72,36,12)}& asN &.0490&.3867&.8623&.9877&1.000&1.000\\
&&Bootstrap& .0563&.4133&.8803&.9907&1.000&1.000\\\hline

\multirow{9}*{ (5, 10)}&& &(.0017)&(.0593)&(.1114)&(.1605)&(.2061)&(.2461)\\
 & \multirow{2}*{(40,40,40)}& asN &.0600 &.5697&.9640&1.000&1.000&1.000\\
&&Bootstrap& .0960&.6790&.9793&1.000&1.000&1.000\\\cline{2-9}
&& &(.0010)&(.0573)&(.1091)&(.1583)&(.2028)&(.2414)\\
& \multirow{2}*{(50,40,30)}& asN &.0500 &.5357 &.9557&1.000&1.000&1.000\\
&&Bootstrap& .0910&.6553&.9717&1.000&1.000&1.000\\\cline{2-9}
&& &(-.0001)&(.0468)&(.0911)&(.1352)&(.1752)&(.2133)\\
& \multirow{2}*{(72,36,12)}& asN &.0493&.4233 &.8793 &.9940 &1.000&1.000\\
&&Bootstrap& .0890&.5557 &.9353 &.9973 &1.000&1.000\\\hline
  \hline \hline 
\end{tabular}
\caption{Size and Power of tests for Example \ref{exam:dependent}(a). }
\label{tab:indKnorm}
\end{table}

\begin{figure}[]
\centering
\begin{tabular}{cc}
\includegraphics[width=0.7\textwidth]{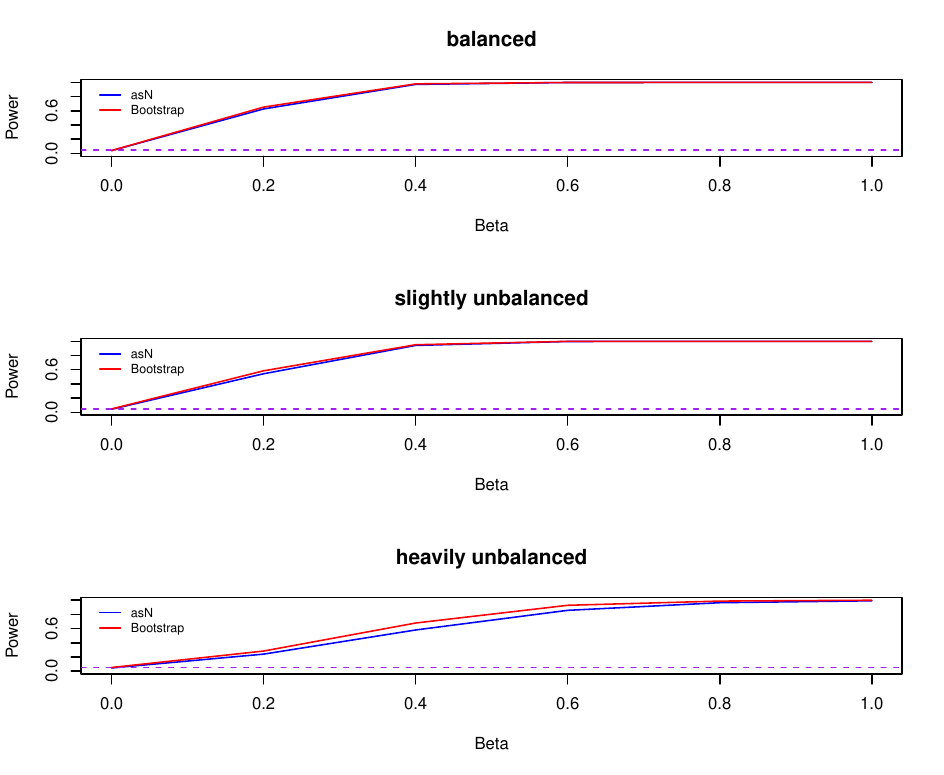}
\end{tabular} 
\caption{Size and power of tests in Example \ref{exam:dependent}(b). Dashed horizontal line is the nominal level 0.05.}
\label{fig:Rexp}
\end{figure}

From Table \ref{tab:indKnorm} and Figure \ref{fig:Rexp}, we observe that  the proposed asN and Bootstrap methods exhibit comparable performance across the various simulation settings when the predictors are dependent. Both tests effectively control the Type I error at the nominal level, demonstrating good size accuracy. 
As the magnitude of the correlation differences increases, the powers of both methods improves markedly in all scenarios.
Notably, under conditions of heavily unbalanced sample sizes ($\bi n=(72, 36, 12)$), the Bootstrap method slightly outperforms the asN approach in terms of power. This modest advantage is consistent with the results from the independent predictor setting examined in Example \ref{exam:independent}, suggesting that the resampling-based method offers greater robustness in the presence of  unequal class sizes in practice.

\begin{example}[Binary Logistic Regression Model]\label{exam:blrm}
In this example, we consider a logistic regression model where the binary response is generated as
\begin{align*}
\log\left\{ \frac{P(Z=1| V)}{P(Z=-1| V)}\right\} = -3 +2 V_1+2V_2+2V_3+3\sin(V_4) +4V_5^2,
\end{align*}
where $V \sim N(\bi 0, \bi \Sigma)$ with $\bi \Sigma =(\rho_{ij})_{p \times p}$ having two scenarios $\rho_{ij} =0$ and  $\rho_{ij} =0.5^{|j-i|}, i \neq j$. 
This model has been applied in He,  Ma and Xu \cite{He19}, Mai and Zou \cite{Mai13}, and Sang and Dang \cite{Sang2024} for ultra-high dimensional feature screening.
Here, we consider the regular setting where the dimensionality of $V$ is $10$. That is,  $V= (V_1, V_2, \cdots, V_{10})^T$.
We let  $X=(V_1, V_2, V_3, V_4, V_5)^T$ be  the five active predictors. In order to evaluate the proposed tests, we choose $Y$ from the non-active predictors.  
\end{example}
We consider two inferential questions in this example. 
\begin{enumerate}[\bf (a).]
\item The first question is to assess whether predictor group $X$ exhibits stronger marginal association with the response than $Y$, that is, to evaluate (\ref{test:H0}). 
We report  $p$-values for asN and Bootstrap tests at two different sample sizes $n=60$ and $n=100$, based on $3,000$ simulation repetitions for each case at significance level $\alpha=0.05$ in Table \ref{tab:blrm}.  The upper portion of the table corresponds to the setting where predictors $X$ and $Y$ are independent, while the lower portion reflects the case where they are dependent. Here, $X$ denotes the set of active predictors, and $Y$ consists of non-active predictors with respect to the binary response variable, satisfying $\rho_1>\rho_2$.
  The results suggest that the dependence structure among the predictors enhances the ability of the tests to correctly identify active predictors. In contrast, when the predictors are independent, as sample size of $n=60$ appears insufficient to achieve reliable detection. This is also discussed in  Example 3.2 of \cite{Sang2024}.

 \begin{table}[h] 
\center

\renewcommand{\arraystretch}{1.1}
\begin{tabular}{ccccccc}
\hline \hline

 $\bi \Sigma =(\rho_{ij})_{p \times p}$&  $X$ 	&  $Y $ &  $ n $&$\hat{\rho}_1-\hat{\rho}_2$ & Bootstrap&asN  \\ \hline  



\multirow{6}*{$\rho_{ij}=0, i \neq j$}&\multirow{6}*{$(V_1, V_2, V_3, V_4, V_5)^T$} &\multirow{2}*{$(V_6, V_7, V_8, V_9, V_{10})^T$}&60&.0361&.0479&.0390\\  
&&&100&.0359&.0084&.0052\\ \cline{3-7}
&&\multirow{2}*{$(V_6, V_7, V_8)^T$}&60&.0357&.0659&.0576\\  
&&&100&.0357&.0133&.0106\\ \cline{3-7}
&&\multirow{2}*{$V_6$}&60&.0358&.1180&.1066\\  
&&&100&.0363&.0402&.0396\\ \cline{1-7}

\multirow{6}*{$\rho_{ij} =0.5^{|j-i|}, i \neq j$}&\multirow{6}*{$(V_1, V_2, V_3, V_4, V_5)^T$} &\multirow{2}*{$(V_6, V_7, V_8, V_9, V_{10})^T$}&60&.0872&.0066&.0050\\  
&&&100&.0884&.0004&.0002\\ \cline{3-7}
&&\multirow{2}*{$(V_6, V_7, V_8)^T$}&60&.0868&.0085&.0075\\  
&&&100&.0875&.0007&.0005\\ \cline{3-7}
&&\multirow{2}*{$V_6$}&60&.0849&.0242&.0252\\  
&&&100&.0851&.0038&.0044\\ \cline{1-7}

\end{tabular}
\caption{$p$-values of the proposed  tests for Example \ref{exam:blrm} (a).}
\label{tab:blrm}
\end{table}

\item The second question is to examine whether adding predictor $Y$ to the baseline group $X$ provides additional information, as discussed in Remark \ref{remark:2}. Let $W = \begin{pmatrix} X \\ Y \end{pmatrix}$, the hypothesis of interest for this question is 
\[
 \mathcal{H}_0:\ \rho_g(W,Z) = \rho_g(X,Z)
\quad \text{versus} \quad
 \mathcal{H}_1:\ \rho_g(W,Z) > \rho_g(X,Z).
\] 
$p$-values for the asN and Bootstrap tests under the same setting as in part (a) are presented  in Table \ref{tab:blrm2}.  The uniformly large $p$-values indicate that $Y$ provides no 
additional marginal information beyond $X$. This is expected, as  $Y$ consists of truly non-active predictors that are 
unrelated to the response. 

In addition, in multinomial regression, a related but distinct
question is typically examined via a likelihood ratio test (MLRT) comparing models $Z\sim X$ and
$Z\sim X+Y$, which evaluates conditional improvement in model fit under a
parametric framework. In contrast, our procedure assesses the added value of $Y$
through changes in categorical Gini correlation in a model-free manner. Nevertheless, for this second research question, we also report likelihood-based 
results to provide practical context and to illustrate how the proposed model-free 
approach relates to commonly used regression-based methods, rather than to directly 
compare their performance. Notably, the MLRT and CGC-based methods are 
in agreement here, with MLRT $p$-values hovering around 0.4--0.5 across all settings. 

%
 \begin{table}[H] 
\center

\renewcommand{\arraystretch}{1.1}
\begin{tabular}{ccccccc}
\hline \hline

 $\bi \Sigma =(\rho_{ij})_{p \times p}$&  $X$ 	&  $Y $ &  $ n $& Bootstrap&asN&MLRT  \\ \hline  



\multirow{6}*{$\rho_{ij}=0, i \neq j$}&\multirow{6}*{$(V_1, V_2, V_3, V_4, V_5)^T$} &\multirow{2}*{$(V_6, V_7, V_8, V_9, V_{10})^T$}&60&.9607&.9690&.4111\\   
&&&100&.9942&.9966&.4477\\ \cline{3-7}
&&\multirow{2}*{$(V_6, V_7, V_8)^T$}&60&.9497&.9574&.4464\\  
&&&100&.9912&.9934&.4645\\ \cline{3-7}
&&\multirow{2}*{$V_6$}&60&.9074&.9177&.4787\\  
&&&100&.9704&.9720&.4808\\ \cline{1-7}

\multirow{6}*{$\rho_{ij} =0.5^{|j-i|}, i \neq j$}&\multirow{6}*{$(V_1, V_2, V_3, V_4, V_5)^T$} &\multirow{2}*{$(V_6, V_7, V_8, V_9, V_{10})^T$}&60&.9946&.9959&.3814\\  
&&&100&.9997&.9999&.4370\\ \cline{3-7}
&&\multirow{2}*{$(V_6, V_7, V_8)^T$}&60&.9921&.9932&.4219\\  
&&&100&.9995&.9997&.4633\\ \cline{3-7}
&&\multirow{2}*{$V_6$}&60&.9747&.9760&.4580\\  
&&&100&.9964&.9965&.4916\\ \cline{1-7}

\end{tabular}
\caption{$p$-values of the proposed  tests for  Example \ref{exam:blrm} (b).}
\label{tab:blrm2}
\end{table}
\end{enumerate}

  
%


\section{Real Data Analysis}\label{sec:realdata}

In this section, we illustrate the proposed CGC-based
predictor comparison procedures using three real-world datasets. The first two datasets
are related to breast cancer: one arises from genomic data, and the other from
medical diagnostics. The third dataset comes from a human activity recognition
study. In each application, two meaningful predictor groups are examined in relation
to a common categorical response variable. These groups are defined either
based on scientific considerations or through data-driven feature selection. Across all datasets, we perform both the group comparison test and the added-value test. As in Example~\ref{exam:blrm}, we additionally report multinomial $p$-values for the added-value setting.
\subsection{TCGA Breast Cancer Gene Expression Data}

The TCGA breast cancer gene expression dataset \cite{Goldman} contains expression measurements
for 17{,}278 genes from 506 patients. Each patient is classified into one of four
intrinsic molecular subtypes (luminal A, luminal B, HER2-enriched, and basal-like), which together define the categorical response variable $Z$. Among the measured genes, the PAM50 gene signature is widely recognized as a gold standard for molecular subtype characterization and prognosis. Accordingly, we define biologically meaningful predictor groups, with $X$ representing the PAM50 genes and $Y$ representing the non-PAM genes. 

We first consider an initial analysis in which $Y$ consists of 48 non-PAM genes. Table~\ref{tab:TCGA_realdata} summarizes the corresponding CGC-based inference results. In the marginal CGC comparison test, $(H_1: \rho_g(X,Z)>\rho_g(Y,Z))$, the estimated CGC difference is large and positive, and both the bootstrap and asN tests yield extremely small $p$-values ($<0.001$), indicating that the PAM50 genes provide substantially more predictive information about breast cancer subtype than the competing gene set.

\begin{table}[ht]
\centering
\renewcommand{\arraystretch}{1.3}
\resizebox{0.95\textwidth}{!}{%
\begin{tabular}{lccccccc}
\hline\hline
Test & $n$ & $K$ & CGC Difference & Bootstrap $p$ & asN $p$ & Multinomial $p$ & Interpretation \\
\hline
$\rho_g(X,Z)-\rho_g(Y,Z)$ 
& 506 & 4 & $0.1926$ & $<0.001$ & $<0.001$ & $-$ & $X$ stronger than $Y$ \\

$\rho_g(W,Z)-\rho_g(X,Z)$
& 506 & 4 & $-0.0541$ & $>0.999$ & $>0.999$ & $>0.99$ & No added marginal value \\

\hline
\multicolumn{8}{l}{\footnotesize $X$: PAM48 genes; $Y$: non-PAM48 genes.} \\
\multicolumn{8}{l}{\footnotesize Distance correlation between $X$ and $Y$: $0.58$.} \\
\hline
\end{tabular}}
\caption{CGC results for the TCGA breast cancer gene expression dataset.}
\label{tab:TCGA_realdata}
\end{table}

In contrast, the added-value CGC test for assessing additional predictive information,
$(H_1: \rho_g(W,Z)-\rho_g(X,Z))$, yields a negative estimated CGC difference, and both inference procedures produce very large $p$-values ($>0.999$). This provides no evidence that augmenting the PAM50 genes with additional non-PAM genes contributes additional predictive information for the response. The negative CGC difference can be explained by the substantial correlation between $X$ and $Y$: when $Y$ contains information already captured by $X$, combining the two predictor groups does not provide further predictive information, and the CGC estimator may even produce a slightly smaller value. These findings are further supported by the multinomial regression likelihood ratio test, which yields a larger $p$-value.

To further explore the behavior of the proposed method in higher dimensional settings, we extend the assessment of additional predictive information by increasing the dimension of the non-PAM predictor group $Y$. Specifically, we consider configurations in which $Y$ contains 100, 200, and 300 non-PAM genes, while keeping $X$ fixed as the PAM48 gene set. This setup allows us to empirically examine how the method behaves when a large number of additional predictors, many of which may be weakly associated or redundant, are included.

The results are summarized in Table~\ref{tab:TCGA_highdimY}. Across all higher dimensional configurations, the estimated CGC difference remains negative and becomes increasingly so as the dimension of $Y$ grows. Both the bootstrap and asN tests consistently yield larger p-values, indicating that the inferential conclusion is stable with respect to the dimension of the added predictor group. At the same time, the distance correlation between $X$ and $Y$ increases substantially as $\dim(Y)$ grows, reflecting stronger redundancy between the PAM genes and larger non-PAM gene panels. This reflects the fact that adding many weakly associated or redundant predictors does not increase the marginal association with $Z$ and can slightly dilute the association captured by $X$ alone. Overall, these results demonstrate that the proposed CGC-based test remains well behaved and interpretable in higher dimensional settings. Even when hundreds of additional predictors are included, the method correctly identifies that the additional variables do not contribute further marginal association beyond that captured by the PAM50 signature. This analysis highlights the practical applicability of the proposed framework to modern higher dimensional genomic data.

\begin{table}[ht]
\centering
\small
\renewcommand{\arraystretch}{1.1}
\begin{tabular}{cccc}
\hline\hline
$\dim(X)$ & $\dim(Y)$ & CGC difference & $\mathrm{dCor}(X,Y)$ \\
\hline
48 & 100 & $-0.0822$ & $0.71$ \\
48 & 200 & $-0.1100$ & $0.83$ \\
48 & 300 & $-0.1288$ & $0.85$ \\
\hline
\multicolumn{4}{l}{\footnotesize Bootstrap and asN $p$-values $>0.999$ for all configurations.} \\
\multicolumn{4}{l}{\footnotesize $X$: PAM48 genes; $Y$: non-PAM genes of increasing dimensionality; $W=(X,Y)$.} \\
\hline
\end{tabular}
\caption{CGC-based added-value analysis for the TCGA breast cancer dataset with increasing dimensionality of the added predictor group $Y$ ($n=506$).}
\label{tab:TCGA_highdimY}
\end{table}

\subsection{Wisconsin Breast Cancer Diagnostic Data}

The Wisconsin Breast Cancer (Diagnostic) dataset \cite{Wolberg1995} consists of $n=569$ observations
obtained from digitized images of fine needle aspirates of breast masses. Each
observation is classified as either benign or malignant, yielding a binary
categorical response variable $Z$. For each observation, 30 numerical predictors
are available, describing characteristics of cell nuclei. These predictors naturally
decompose into feature groups according to measurement type: for each of ten basic
characteristics (such as radius, texture, and perimeter), the dataset provides the
mean value, the standard error (SE), and the worst (largest) value.

To identify the most informative predictors, we first computed Random Forest
feature importance across all 30 predictors. As shown in Figure~\ref{fig:WISC_top15}, the top 10 features, dominated by
Worst and Mean measurements of radius, perimeter, area, and concave points, were
selected as $X$, while the remaining 20 features were assigned to $Y$. This
selection allows us to test whether additional features contribute predictive
information beyond the top features.

Table~\ref{tab:WISC_realdata_top10} summarizes the results of CGC-based and
logisitic regression analyses. The marginal CGC comparison indicates that the top 10 features capture significantly more predictive information
than the remaining features, with both bootstrap and asN $p$-values
being extremely small ($<0.001$). The added-value CGC test shows that
including the remaining features does not improve predictive information
($p$-values $>0.999$).  A ridge-regularized logistic regression comparison
for the added-value setting also yields a large $p$-value ($p=0.2575$), indicating
that the remaining features provide no statistically significant additional predictive information beyond the top 10 predictors.

Overall, these results demonstrate that the top 10 features already saturate the
nonparametric predictive signal for malignancy, and adding remaining features provides
little added value.

\begin{table}[ht]
\centering
\renewcommand{\arraystretch}{1.3}
\resizebox{0.95\textwidth}{!}{%
\begin{tabular}{lccccccc}
\hline\hline
Test & $n$ & $K$ & CGC Difference & Bootstrap $p$ & asN $p$ & Multinomial $p$ & Interpretation \\
\hline
$\rho_g(X,Z)-\rho_g(Y,Z)$ 
& 569 & 2 & $0.1562$ & $<0.001$ & $<0.001$ & $-$ & $X$ stronger than $Y$ \\

$\rho_g(W,Z)-\rho_g(X,Z)$
& 569 & 2 & $-0.1062$ & $>0.999$ & $>0.999$ & $0.2575$ & No added marginal value \\

\hline
\multicolumn{8}{l}{\footnotesize $X$: top 10 RF features; $Y$: remaining 20 features.} \\
\multicolumn{8}{l}{\footnotesize Distance correlation between $X$ and $Y$: $0.8314$.} \\
\hline
\end{tabular}}
\caption{CGC results for the Wisconsin Breast Cancer dataset using top 10 Random Forest features.}
\label{tab:WISC_realdata_top10}
\end{table}

\begin{figure}[ht]
\centering
\includegraphics[width=0.75\textwidth]{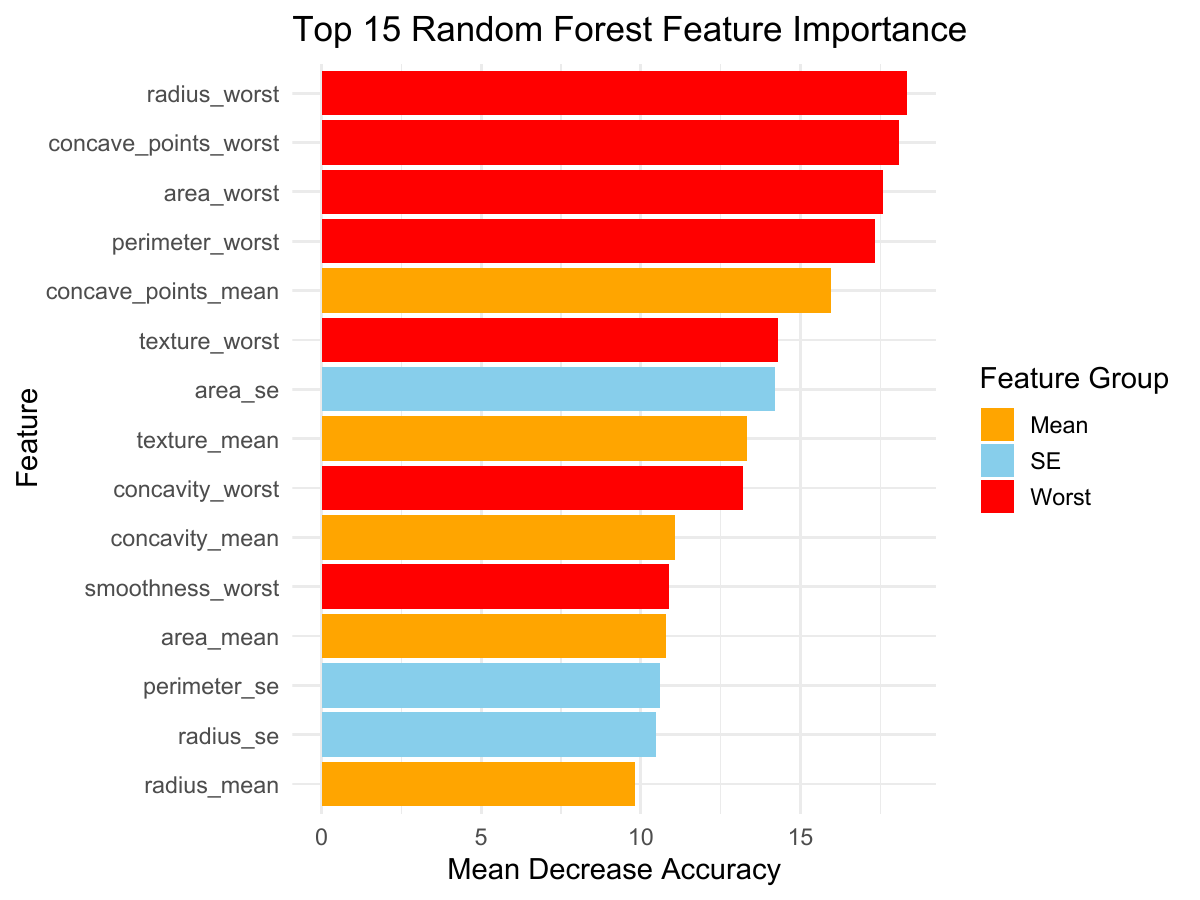}
\caption{Top 15 Random Forest feature importances for the Wisconsin Breast Cancer dataset,
colored by feature group (SE, Mean, Worst). Worst features dominate, followed by
Mean features, while SE features contribute less.}
\label{fig:WISC_top15}
\end{figure}

\subsection{Human Activity Recognition (HAR) Data}

The Human Activity Recognition (HAR) dataset \cite{HARdataset,Anguita2013} contains smartphone sensor measurements from 10{,}299 observations, each labeled as one of six physical activities. Each observation includes 561 numerical features derived from accelerometer and gyroscope signals. These features naturally form two predefined predictor groups based on sensor modality: accelerometer-related features ($X$) and gyroscope-related features ($Y$). To facilitate efficient and balanced analysis across activity classes, we apply a stratified subsampling procedure to obtain a representative subset of $n=1200$ observations.


As shown in Table~\ref{tab:HAR_realdata}, both the bootstrap and asN results ($p<0.001$) indicate that accelerometer features provide significantly stronger predictive information for activity labels than gyroscope features. Furthermore, the assessment of additional predictive value shows no evidence that combining the two sensor modalities yields further improvement in predicting the response. Both conclusions are consistent with expectations. A multinomial likelihood ratio test (MLRT), performed principal component scores also yields a large $p$-value ($0.9692$) consistent with the CGC added-value results, indicating that the gyroscope features contribute little additional predictive information beyond the accelerometer features.

These results suggest that the accelerometer group already captures the dominant predictive structure for activity classification, while the gyroscope features primarily provide redundant information.


\begin{table}[ht]
\centering
\renewcommand{\arraystretch}{1.3}
\resizebox{0.95\textwidth}{!}{%
\begin{tabular}{lccccccc}
\hline\hline
Test & $n$ & $K$ & CGC Difference & Bootstrap $p$ & asN $p$ & MLRT $p$ & Interpretation \\
\hline
$\rho_g(X,Z)-\rho_g(Y,Z)$ 
& 1200 & 6 & $0.0979$ & $<0.001$ & $<0.001$ & $-$ & $X$ stronger than $Y$ \\

$\rho_g((W),Z)-\rho_g(X,Z)$
& 1200 & 6 & $-0.0349$ & $>0.999$ & $>0.999$ & $0.9692$ & No added marginal value \\

\hline
\multicolumn{8}{l}{\footnotesize $X$: accelerometer-related features; $Y$: gyroscope-related features.} \\
\multicolumn{8}{l}{\footnotesize Distance correlation between $X$ and $Y$: $0.97$.} \\
\hline
\end{tabular}}
\caption{Results for the HAR dataset based on a stratified subsample ($n=1200$).}
\label{tab:HAR_realdata}
\end{table}

\subsection{Summary of Real Data Findings}

Across the TCGA breast cancer, Wisconsin Breast Cancer Diagnostic, and Human Activity Recognition datasets, the proposed CGC-based procedures consistently identify predictor groups that provide stronger predictive information for the categorical response, even in the presence of substantial correlation between groups. In all three applications, the comparison between predictor groups yields clear and statistically significant results that align with domain knowledge. 

For example, in the TCGA breast cancer dataset, the PAM50 gene signature emerges as the dominant predictor of molecular subtypes, confirming its established role in subtype characterization. In the Wisconsin Breast Cancer Diagnostic dataset, the top 10 Random Forest-selected features, dominated by the Worst and Mean measurements of radius, perimeter, area, and concave points, exhibit stronger association with malignancy status than the remaining features. In particular, the Worst measurements capture extreme cellular characteristics that are known to be clinically relevant, reflecting patterns aligned with tumor aggressiveness.  In the Human Activity Recognition dataset, sensor modalities capturing movement and orientation provide more predictive information than secondary sensor groups, reflecting the known relevance of these signals in activity classification. Across these datasets, both the proposed CGC-based methods and the MLRT assessment of additional predictive information show no evidence that augmenting a dominant predictor group $X$ with a highly correlated group $Y$ contributes further predictive information for the response, reinforcing the importance of identifying non-redundant predictive information.


\subsection{Computational Considerations}

The proposed methodology requires the computation of pairwise distances among $n$ observations, resulting in a computational complexity of $O(n^2)$. Specifically, the CGC statistic is a U-statistic whose kernel depends on two observations at a time, requiring evaluation over $\binom{n}{2}$ distinct pairs. This quadratic cost is intrinsic to distance-based dependence measures and is shared by related methods such as distance correlation \cite{Szekely2007}, energy statistics \cite{Ramos2023}, and kernel-based two-sample tests. In practice, this cost is manageable for moderate to large sample sizes (e.g., $n \le 10{,}000$) using standard computing resources. For very large datasets, scalability can be improved through several strategies, including class-stratified distance computations within the bootstrap procedure and the use of approximate distance techniques such as random projections or subsampling. In addition, parallel computing can be effectively utilized to distribute pairwise distance calculations across multiple cores. These approaches substantially reduce computation time while preserving the validity of the proposed CGC-based inference.

\section{Conclusions and Future work}\label{sec:conclusion}

In this paper, we introduce a hypothesis test for comparing two categorical Gini correlations of two numerical predictors with respect to a common categorical response. The paper focuses on comparing two population-level categorical Gini correlations. We develop a formal test to assess whether a significant difference exists between them. The proposed test statistic is derived, and its asymptotic distribution is studied under both the null and alternative hypotheses. In addition, we propose a nonparametric bootstrap method as an alternative approach for inference, offering a flexible procedure to assess differences in categorical Gini correlations.

Following the establishment of theoretical results for the inference methods, we conducted extensive simulation studies to verify and validate these findings, considering both univariate and multivariate cases. Simulation studies suggest that the proposed methods work effectively under both independent and dependent correlation structures. 
 After confirming the theoretical results, we applied the proposed methods in a practical scenario. The proposed procedures provide a useful framework for evaluating and comparing numerical feature groups in classification settings with categorical responses. Thanks to the correlation's ability to work with multidimensional data, the proposed methods turned out to be useful for a wide range of real-world applications.

The categorical Gini covariance has been generalized to a RKHS in \cite{Zhang2021} with a Mercer kernel.
\begin{align*} 
\mbox{gCov}( X,Z; d_\kappa) &=\E  \{ d_\kappa(X_1,X_2)\}-\sum_{k=1}^Kp_k\E \{d_\kappa(X_1^{(k)}, X_2^{(k)})\},
\end{align*}
where $d_\kappa(x_1, x_2) = \sqrt{\kappa( x_1, x_1)+\kappa( x_2, x_2)-2\kappa(x_1, x_2)}$, the distance in the feature space induced by positive definite kernel $\kappa$. 
In the future, we will compare the significance of two covariates to the same categorical response in RKHS by developing statistical inference of kernel categorical Gini correlation. 
A potential limitation of the proposed methods is that, although predictors may have arbitrary and unequal dimensions, differences in scale or distributional properties could influence the CGC values, potentially affecting the validity of the comparisons. Future research could investigate strategies to mitigate these effects and ensure more robust inference. Another natural extension is to generalize the proposed tests beyond pairwise group comparison to multiple or hierarchically structured predictor groups. Such developments would be particularly relevant for classification problems  involving nested feature representations, multi-modal data, or pathway-based groupings.

Exploring the modified estimator proposed by Sang and Dang (2024) \cite{Sang2024} offers another promising direction, as it may admit a simpler asymptotic distribution. While the present paper focuses on inference for the original CGC as defined in \cite{Dang2021}, incorporating this modified estimator could yield new theoretical insights and potentially simplify practical applications.

Finally, another important avenue for future research is to extend the proposed 
CGC-based framework by incorporating copula-based modeling. The current framework 
treats the joint distribution of predictor groups nonparametrically, whereas copulas 
provide a flexible mechanism for explicitly characterizing dependence between $X$ and 
$Y$, including nonlinear and tail dependence structures that may influence CGC-based 
inference. Such extensions are particularly relevant in biomedical applications, where 
complex dependence patterns commonly arise. Recent work on copula-based methods for 
joint tail risk in paired biomarkers \cite{Aich4}, fusion of clinical and genomic risk 
scores \cite{Aich3}, feature selection \cite{Aich1}, and imbalanced classification 
\cite{Aich2} provides a natural foundation for this direction, and suggests that 
copula-enhanced CGC inference may yield more robust and interpretable assessments of 
independent predictive contribution in complex data settings.

\section{Appendix}

\begin{proof}[Proof of Theorem \ref{wilkrho2}]

We first assume the GMDs for $F$ and $G$ are known. Define
\[
\hat{\rho}_1 = \frac{\mathrm{gCov}_n(X,Z)}{\Delta_F}, \quad 
\hat{\rho}_2 = \frac{\mathrm{gCov}_n(Y,Z)}{\Delta_G},
\]
and consider their difference
\[
\hat{D}_n = \hat{\rho}_1 - \hat{\rho}_2 = \sum_{k \neq l} \hat{p}_k \hat{p}_l \hat{D}_{kl},
\]
where
\[
\hat{D}_{kl} = \frac{1}{n_k^2 n_l^2} \sum_{i_1,i_2=1}^{n_k} \sum_{j_1,j_2=1}^{n_l} 
h\Big( (X^{(k)}_{i_1}, Y^{(k)}_{i_1}), (X^{(k)}_{i_2}, Y^{(k)}_{i_2}); 
(X^{(l)}_{j_1}, Y^{(l)}_{j_1}), (X^{(l)}_{j_2}, Y^{(l)}_{j_2}) \Big),
\]
with kernel
\begin{align*}
h_{kl} &:= \frac{1}{2 \Delta_F} \Big[ \|X^{(k)}_{i_1}-X^{(l)}_{j_1}\| + \|X^{(k)}_{i_2}-X^{(l)}_{j_2}\| - \|X^{(k)}_{i_1}-X^{(k)}_{i_2}\| - \|X^{(l)}_{j_1}-X^{(l)}_{j_2}\| \Big] \\
&\quad - \frac{1}{2 \Delta_G} \Big[ \|Y^{(k)}_{i_1}-Y^{(l)}_{j_1}\| + \|Y^{(k)}_{i_2}-Y^{(l)}_{j_2}\| - \|Y^{(k)}_{i_1}-Y^{(k)}_{i_2}\| - \|Y^{(l)}_{j_1}-Y^{(l)}_{j_2}\| \Big].
\end{align*}

Let $\tilde{h}_{kl} = h_{kl} - \mathbb{E}[h_{kl}]$ be the centered kernel, and define its first-order projections:
\begin{align*}
\tilde{h}^{10}_{kl}(x^{(k)},y^{(k)}) &= \frac{1}{2 \Delta_F} \Big[ \mathbb{E}\|x^{(k)} - X_1^{(l)}\| - \mathbb{E}\|x^{(k)} - X_2^{(k)}\| - \Delta_{kl,F} + \Delta_{k,F} \Big] \\
&\quad - \frac{1}{2 \Delta_G} \Big[ \mathbb{E}\|y^{(k)} - Y_1^{(l)}\| - \mathbb{E}\|y^{(k)} - Y_2^{(k)}\| - \Delta_{kl,G} + \Delta_{k,G} \Big], \\
\tilde{h}^{01}_{kl}(x^{(l)},y^{(l)}) &= \frac{1}{2 \Delta_F} \Big[ \mathbb{E}\|x^{(l)} - X_1^{(k)}\| - \mathbb{E}\|x^{(l)} - X_2^{(l)}\| - \Delta_{kl,F} + \Delta_{l,F} \Big] \\
&\quad - \frac{1}{2 \Delta_G} \Big[ \mathbb{E}\|y^{(l)} - Y_1^{(k)}\| - \mathbb{E}\|y^{(l)} - Y_2^{(l)}\| - \Delta_{kl,G} + \Delta_{l,G} \Big].
\end{align*}

The first-order projection of $\hat{D}_{kl}$ is
\[
\hat{D}^{(1)}_{kl} = 2 \Bigg[ \frac{1}{n_k} \sum_{i=1}^{n_k} \tilde{h}^{10}_{kl}(x^{(k)}_i, y^{(k)}_i) + \frac{1}{n_l} \sum_{i=1}^{n_l} \tilde{h}^{01}_{kl}(x^{(l)}_i, y^{(l)}_i) \Bigg],
\]
and hence
\[
\hat{D}^{(1)} = \sum_{k \neq l} \hat{p}_k \hat{p}_l \hat{D}_{kl}^{(1)} = \frac{4}{n^2} \sum_{k \neq l} n_l \sum_{i=1}^{n_k} \tilde{h}^{10}_{kl}(x^{(k)}_i, y^{(k)}_i).
\]

Under $\mathcal{H}_0$, $\mathbb{E}[\tilde{h}^{10}_{kl}] = \mathbb{E}[\tilde{h}^{01}_{kl}] = 0$, and the variance of the first-order projection is
\[
\sigma_0^2 = \operatorname{Var}(\hat{D}^{(1)}) = 16 \sum_{k \neq l} (p_k^2 p_l + p_l^2 p_k) \operatorname{Var}(\tilde{h}^{10}_{kl}(X^{(k)}, Y^{(k)})).
\]

By condition \textbf{C4}, the kernel $h_{kl}$ varies non-trivially, so $\operatorname{Var}(\tilde{h}^{10}_{kl}(X^{(k)}, Y^{(k)})) > 0$ for at least one pair $k \neq l$. Therefore, $\sigma_0^2 > 0$, guaranteeing that the U-statistic is non-degenerate. The central limit theorem for U-statistics then gives
\[
\sqrt{n} \, \hat{D}_n = \sqrt{n} \sum_{k \neq l} \hat{p}_k \hat{p}_l \hat{D}_{kl}^{(1)} + o_P(1) \stackrel{d}{\longrightarrow} \mathcal{N}(0, \sigma_0^2).
\]

When the GMDs are unknown, define
\[
\hat{\rho}_1 = \frac{\mathrm{gCov}_n(X,Z)}{\widehat{\Delta}_F}, \quad 
\hat{\rho}_2 = \frac{\mathrm{gCov}_n(Y,Z)}{\widehat{\Delta}_G}.
\]

By the almost sure convergence of U-statistics \cite{Serfling1980},
\[
\widehat{\Delta}_F = \Delta_F + o_P(1), \quad \widehat{\Delta}_G = \Delta_G + o_P(1).
\]

Then we can write
\[
\sqrt{n} \Bigg( \frac{\hat{\rho}_1 \widehat{\Delta}_F}{\Delta_F} - \frac{\hat{\rho}_2 \widehat{\Delta}_G}{\Delta_G} \Bigg)
= \sqrt{n} \Bigg( (\hat{\rho}_1 - \hat{\rho}_2) + \hat{\rho}_1 \frac{o_P(1)}{\Delta_F} - \hat{\rho}_2 \frac{o_P(1)}{\Delta_G} \Bigg),
\]
where the extra term
\[
\hat{\rho}_1 \frac{o_P(1)}{\Delta_F} - \hat{\rho}_2 \frac{o_P(1)}{\Delta_G} \to 0 \quad \text{in probability}.
\]

By Slutsky's theorem, we finally obtain
\[
\sqrt{n} (\hat{\rho}_1 - \hat{\rho}_2) \stackrel{d}{\longrightarrow} \mathcal{N}\Big(0, 16 \sum_{k \neq l} (p_k^2 p_l + p_l^2 p_k) \sigma^2_{kl} \Big),
\]
where
\[
\sigma^2_{kl} = \operatorname{Var}(\tilde{h}^{10}_{kl}(X^{(k)}, Y^{(k)})) = \operatorname{Var}(\tilde{h}^{01}_{kl}(X^{(l)}, Y^{(l)}))
\]
for at least one pair \(k \neq l\).

\end{proof}

\begin{proof}[Proof of Theorem \ref{wilkrho3}] 
Under the alternative hypothesis $\mathcal{H}_1: \rho_1 \neq \rho_2$, we use the same first-order projection representation of $\hat{D}_n$ as in Theorem~\ref{wilkrho2}: 
\[
\hat{D}^{(1)} = \sum_{k \neq l} \hat{p}_k \hat{p}_l \hat{D}_{kl}^{(1)}
= \frac{4}{n^2} \sum_{k \neq l} n_l \sum_{i=1}^{n_k} \tilde{h}^{10}_{kl}(x^{(k)}_i, y^{(k)}_i),
\]
where $\tilde{h}^{10}_{kl}$ are the first-order projections of the centered kernel functions $h_{kl}$, defined in Theorem~\ref{wilkrho2}.

We will show that at least one  $\sigma^2_{kl}=\textrm{Var}\bigg(\tilde{h}^{10}_{kl}\big( X^{(k)}, Y^{(k)}\big)\bigg)$ is not zero when $\rho_1\neq\rho_2$


If all the variances are zero, then $\tilde{h}^{10}_{kl}\big( X^{(k)}, Y^{(k)}\big)=\E \tilde{h}^{10}_{kl}\big( X^{(k)}, Y^{(k)}\big)=0$ and $\tilde{h}^{01}_{kl}\big( X^{(l)}, Y^{(l)}\big)=\E \tilde{h}^{01}_{kl}\big( X^{(l)}, Y^{(l)}\big)=0$ for all $k \neq l$. This implies 
\begin{align*}
\Delta_G (\Delta_{k, F}-\Delta_{kl, F})=\Delta_F (\Delta_{k, G}-\Delta_{kl, G}), \ \ \  \Delta_G (\Delta_{l, F}-\Delta_{kl, F})=\Delta_F (\Delta_{l,G}-\Delta_{kl, G})
\end{align*}
and hence 
\begin{align*}
\Delta_G \sum_{1\leq k <l \leq K}p_kp_l (\Delta_{k, F}-\Delta_{kl, F}+\Delta_{l, F}-\Delta_{kl, F})=\Delta_F \sum_{1\leq k <l \leq K}p_kp_l (\Delta_{k,G}-\Delta_{kl, G}+\Delta_{l,G}-\Delta_{kl, G}), 
\end{align*}
and hence 
$\rho_1=\rho_2$.
Therefore, when $\rho_1 \neq \rho_2$, we have 
\begin{align}\label{eqn:norm}
 \sqrt{n}  \bigg((\hat{\rho}_1  - \hat{\rho}_2 ) - (\rho_1 - \rho_2)\bigg)\stackrel{d}{\rightarrow} \mathcal{N}(0,16 \sum_{k \neq l}^K p^2_lp_k\sigma^2_{kl}).
 \end{align}
 
 When both the two GMDs are unknown,  $\hat{\rho}_1=\dfrac{\textrm{gCov}_n(X, Z)}{\widehat{\Delta}_F}$,  $\hat{\rho}_2=\dfrac{\textrm{gCov}_n(Y, Z)}{\widehat{\Delta}_G}$. 
 
 By (\ref{eqn:norm}), 
\begin{align}\label{eqn:norm2}
 \sqrt{n}  \bigg(\Big(\dfrac{\hat{\rho}_1 \widehat{\Delta}_F}{\Delta_F} - \dfrac{\hat{\rho}_2 \widehat{\Delta}_G}{\Delta_G}\Big) - (\rho_1 - \rho_2)\bigg)\stackrel{d}{\rightarrow} \mathcal{N}(0,16 \sum_{k \neq l}^K p^2_lp_k\sigma^2_{kl}).
 \end{align}

As $n \to \infty$,  by the almost sure behaviors of $U$-statistics \cite{Serfling1980}, we have $\widehat{\Delta}_F=\Delta_F+o(1)$ and $\widehat{\Delta}_G=\Delta_G+o(1)$. 
Then (\ref{eqn:norm2}) can be written as 
$$ \sqrt{n}  \bigg(\Big(\hat{\rho}_1\dfrac{ o(1)}{\Delta_F} -\hat{\rho}_2 \dfrac{ o(1)}{\Delta_G}\Big) +(\hat{\rho}_1-\hat{\rho}_2)- (\rho_1 - \rho_2)\bigg)\stackrel{d}{\rightarrow} \mathcal{N}(0,16 \sum_{k \neq l}^K p^2_lp_k\sigma^2_{kl}),$$
where $\Big(\hat{\rho}_1\dfrac{ o(1)}{\Delta_F} -\hat{\rho}_2 \dfrac{ o(1)}{\Delta_G}\Big) \to 0$ a.s.. 
By Slutsky theorem, and condition \textbf{C}1, we have 
\begin{align}\label{sig:h1}
\sqrt{n}  \bigg((\hat{\rho}_1  - \hat{\rho}_2)  - (\rho_1 - \rho_2)\bigg)\stackrel{d}{\rightarrow} \mathcal{N}(0,16 \sum_{k \neq l}^K p^2_lp_k\sigma^2_{kl}).
\end{align}

\end{proof}

\section{Acknowledgements}
We would like to thank the editor and referees for their valuable guidance and thoughtful suggestions, which have greatly improved the clarity and quality of this paper.

\textbf{Funding}: The research of Sameera Hewage was supported by the United States National Institutes of Health (NIH) under grant P20GM103434 through the West Virginia IDeA Network of Biomedical Research Excellence.

\end{document}